\documentclass[12pt]{article}

\usepackage{geometry}
\usepackage{cite}
\usepackage{paralist}
\usepackage{graphicx}
\usepackage{subfigure}
\usepackage{amssymb}
\usepackage{amsmath}
\allowdisplaybreaks[1]

\usepackage{amsthm}

\newtheorem{theorem}{Theorem}
\newtheorem{lemma}[theorem]{Lemma}

\newtheorem{claim}[theorem]{Claim}

\makeatletter
\def\@endtheorem{\endtrivlist}
\makeatother

\newcounter{rrule}
\newenvironment{rrule}{\refstepcounter{rrule}\par\smallskip\noindent
\textbf{(\arabic{rrule})}\quad}{}
\newcommand{\currentrule}{\arabic{rrule}}

\newcommand{\row}[1]{\{#1\}\times [1,m]}
\newcommand{\col}[1]{[1,n] \times \{#1\}}
\newcommand{\vcut}[3]{[#1,#2] \times \{#3\}}
\newcommand{\hcut}[3]{\{#1\} \times [#2,#3]}
\begin{document}

\title{An FPT algorithm for orthogonal buttons and scissors}
\author{Dekel Tsur%
\thanks{Ben-Gurion University of the Negev.
Email: \texttt{dekelts@cs.bgu.ac.il}}}
\date{}
\maketitle

\begin{abstract}
We study the puzzle game Buttons and Scissors in which
the goal is to remove all buttons from an $n\times m$ grid by a series of
horizontal and vertical cuts.
We show that the corresponding parameterized problem has an algorithm with time
complexity $2^{O(k^2 \log k)} (n+m)^{O(1)}$,
where $k$ is an upper bound on the number of cuts.
\end{abstract}

\paragraph{Keywords} Combinatorial puzzles;
Parameterized complexity.

\section{Introduction}
In the Buttons and Scissor puzzle one is given an $n \times m$ grid,
where some of the cells of the grid contain \emph{buttons}.
Each button has a color.
The goal is to remove all buttons from the grid by applying \emph{cuts}.
A cut is a sequence $C = ((i_1,j_1),\ldots,(i_p,j_p))$ of grid cells with the
following properties.
\begin{enumerate}
\item
The cell $(i_1,j_1)$ and $(i_p,j_p)$ contain buttons.
\item
All the buttons in the cells of $C$ have the same color.
\item
$C$ has one of the following forms.
\begin{enumerate}
\item
$i_1 = \cdots = i_p$ and $j_{l+1} = j_l+1$ for all $l < p$.
\item
$j_1 = \cdots = j_p$ and $i_{l+1} = i_l+1$ for all $l < p$.
\item
$i_{l+1} = i_l+1$ and $j_{l+1} = j_l+1$ for all $l < p$.
\end{enumerate}
\end{enumerate}
The cuts of the first, second, and third form above are called
\emph{horizontal cuts}, \emph{vertical cuts}, and \emph{diagonal cuts},
respectively.
The \emph{application} of a cut $C$ deletes all the buttons in the cells of $C$.
A horizontal cut $C = ((i,j_1),\ldots,(i,j_p))$ will also be denoted
$\hcut{i}{j_1}{j_p}$ and
a vertical cut $C = ((i_1,j),\ldots,(i_p,j))$ will also be denoted
$\vcut{i_1}{i_p}{j}$.
See Figure~\ref{fig:example} for an example.

\begin{figure}
\centering
\includegraphics[scale=1.2]{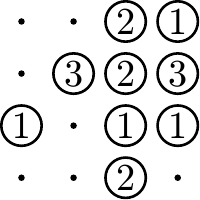}
\caption{An example of the Buttons and Scissor puzzle. In this example,
the buttons can be removed by applying four cuts:
a horizontal cut $\hcut{3}{1}{3}$, a vertical cut $\vcut{1}{4}{3}$,
a horizontal cut $\hcut{2}{2}{4}$, and a verticul cut
$\vcut{1}{3}{4}$.\label{fig:example}}
\end{figure}

The Buttons and Scissors puzzle can be formulated as a decision problem
as follows.
The input is an $n\times m$ matrix $B$ representing the buttons and an integer
$k$, and the goal is to decide whether all buttons can be removed using at most
$k$ cuts.
The matrix $B$ represents the buttons as follows.
If $B[i,j] = 0$ then there is no button at cell $(i,j)$.
If $B[i,j] = c > 0$ then there is a button with color $c$ at cell $(i,j)$.
The \emph{orthogonal buttons and scissor problem} is a variant of
the buttons and scissors problem in which only vertical and horizontal
cuts are allowed.

Both the buttons and scissors problem and the orthogonal button and scissor
problem are NP-hard~\cite{gregg2015buttons}.
Some variants of this problem were studied in~\cite{burke2015single}.

Agrawal et al.~\cite{agrawal2016kernelizing} claimed to give an FTP
algorithm for the orthogonal buttons and scissors problem.
However, there is an error in their algorithm.
In~\cite{agrawal2016kernelizing}, a row block is defined to be a maximal
interval $[a,b]$ of rows of $B$ such that all the non-zeros rows in the block
are identical.
It is claimed that if the number of row blocks is at least $2k+1$ then
$(B,k)$ is a no instance.
However, this is not true. Suppose that $B$ is an $n\times 2$ matrix defined
as follows.
For odd $i$, $B[i,1] = 1$ and $B[i,2] = 0$.
For even $i$, $B[i,1] = 0$ and $B[i,2] = 2$.
The number of row blocks is $n/2$ while all the buttons in $B$ can be removed
using two vertical cuts.

In this paper we give an algorithm for the orthogonal buttons and scissors
problem whose time complexity is $2^{O(k^2 \log k)} (n+m)^{O(1)}$.

\section{The algorithm}
Let $(B,k)$ be an instance of the orthogonal buttons and scissors problem.
A \emph{solution} for $(B,k)$ is a sequence of at most $k$ cuts whose
application remove all the buttons in $B$.
We say that a cut $C$ is \emph{contained} in row $i$ (resp., column $j$)
if $C \subseteq \row{i}$ (resp., $C \subseteq \col{i}$).
A cut $C$ \emph{touches} row $i$ if $C \cap (\row{i}) \neq \emptyset$.
A row or column of $B$ is called \emph{heavy} if it contains at least $k+1$
buttons, and \emph{light} otherwise.
A button is \emph{heavy} if it is in some heavy row or column.
Otherwise, the button is \emph{light}.

For a set of rows $X$ and a set of column $Y$, $B[X,Y]$ is a sub-matrix of $B$
containing the elements $B[x,y]$ for $x\in X$ and $y\in Y$.
A \emph{row block} is a set $X$ of consecutive row indices such that
for every column $j$, all the buttons in $B[X,\{j\}]$ have the same color.
For a row block $X$, let $J_X$ be a set containing every column index $j$
such that $B[X,\{j\}]$ contains at least one button,
and let $J'_X$ be a set containing every column index $j$ such that
$B[X,\{j\}]$ contains between $1$ and $k+1$ buttons.

We now present our algorithm.
The algorithm first repeatedly applies the reduction rules given below.
When no reduction rule is applicable, the algorithm solves the reduced instance.

\begin{rrule}
If there are at more than $k$ heavy rows or more than $k$ heavy columns,
return no.\label{rule:heavy-rows}
\end{rrule}
\begin{lemma}
Rule~(\currentrule) is safe.
\end{lemma}
\begin{proof}
Suppose that $(B,k)$ is yes instance and let $\mathcal{S}$ be a solution for
$(B,k)$.
If no cut of $\mathcal{S}$ is contained in row~$i$ then
row~$i$ is light since a cut in $\mathcal{S}$ can remove at most one button
from row~$i$ of $B$.
Therefore, the number of heavy rows is at most $|\mathcal{S}| \leq k$.
Similarly, the number of heavy rows is at most $k$.
Therefore, Rule~(\currentrule) is safe.
\end{proof}

\begin{rrule}
If there are more than $k^2$ light buttons, return no.\label{rule:light-buttons}
\end{rrule}
\begin{lemma}
Rule~(\currentrule) is safe.
\end{lemma}
\begin{proof}
Suppose that $(B,k)$ is yes instance and let $\mathcal{S}$ be a solution for
$(B,k)$.
Let $C \in \mathcal{S}$ be a horizontal cut that is contained in row $i$.
If row $i$ is heavy then $C$ does not remove light buttons.
Otherwise, $C$ removes at most $k$ light buttons.
Similarly, a vertical cut $C \in \mathcal{S}$ removes at most $k$ light buttons.
Since every button must be removed by a cut in $\mathcal{S}$, it follows that
there are at most $k^2$ light buttons.
\end{proof}

\begin{rrule}
If row~$i$ (resp., column~$j$) does not contain buttons, delete row~$i$
(resp., column~$j$) from $B$.\label{rule:empty}
\end{rrule}

The following reduction rule is adapted from Rule~2
in~\cite{agrawal2016kernelizing}.
For completeness, we give a full proof of the safeness of the rule.

\begin{rrule}\label{rule:row-block}
Suppose that there is row block $X = [a,b]$ and an index $i_0 \in X$ such that
for every column $j$, either (1) $B[X,\{j\}]$ does not contain buttons,
or (2) there are at least $k$ buttons in $B[[a,i_0-1],\{j\}]$ and at least
$k$ buttons in $B[[i_0+1,b],\{j\}]$.
Return the instance $(B',k)$, where $B$ is a matrix obtained from $B$
by deleting all the buttons in row $i_0$ of $B$.
\end{rrule}

\begin{lemma}
Rule~(\currentrule) is safe.
\end{lemma}
\begin{proof}
We need to show that $(B,k)$ is a yes instance if and only if $(B',k)$
is a yes instance.
To prove the first direction we use following claim.

\begin{claim}\label{clm:S}
If $(B,k)$ is a yes instance, there is a solution $\mathcal{S}$ for $(B,k)$
such that no cut of $\mathcal{S}$ is contained in a row of $X$.
Additionally, there is no vertical cut in $\mathcal{S}$ with endpoint
$(i_0,j)$ for some column $j$.
\end{claim}
\begin{proof}
Suppose that $(B,k)$ is yes instance, and let $\mathcal{S}_0$ be a solution for
$(B,k)$.
We construct a new solution $\mathcal{S}$ by generating a cut $C'$ for every
cut $C$ in $\mathcal{S}_0$.
If $C$ is cut in $\mathcal{S}_0$ that does not touch any row in $X$ then
the corresponding cut in $\mathcal{S}$ is $C' = C$.
If $C$ is a horizontal cut in $\mathcal{S}_0$ that is contained in a row in $X$
then the corresponding cut in $\mathcal{S}$ is $C' = \emptyset$.

Let $j$ be some column.
If $B[X,\{j\}]$ does not contain buttons then for every vertical
cut $C$ in $\mathcal{S}_0$ such that $C$ is contained in column $j$ and $C$ touches at least one
row of $X$,
the corresponding cut in $\mathcal{S}$ is $C' = C$.

Now suppose that $B[X,\{j\}]$ contains buttons.
Due to the condition of Rule~(\currentrule),
$B[X,\{j\}]$ contains at least $2k$ buttons.
Therefore, there is an index $i' \in X$ such that $B[i',j] \neq 0$ and there is
no cut in $\mathcal{S}_0$ that is contained in row $i'$.
Therefore, the cut in $\mathcal{S}_0$ that removes the button $B[i',j]$ is
a vertical cut.

Let $C^* = \vcut{a^*}{b^*}{j}$ be the last vertical cut in $\mathcal{S}_0$
such that the application of $C^*$ (after the application of the preceding
cuts in $\mathcal{S}_0$) removes at least one button in $B[X,\{j\}]$.
Note that $C^*$ exists since $B[i',j]$ is removed by a vertical cut.
If $C$ is a vertical cut in $\mathcal{S}_0$ that is contained in column $j$,
touches at least one row of $X$, and appears after $C^*$ in $\mathcal{S}_0$,
then the corresponding cut in $\mathcal{S}$ is $C' = C$.
Now, let $C = \vcut{a'}{b'}{j}$ be a vertical cut in $\mathcal{S}_0$ that
touches at least one row of $X$ and appears before $C^*$ in $\mathcal{S}_0$.
The interval $[a',b']$ cannot contain the interval $[a,b]$ otherwise we get a
contradiction to the definition of $C^*$.
Therefore, there are three possible cases.
\begin{enumerate}
\item
$a' \geq a$ and $b' \leq b$, and at least one inequality is strict.
\item
$a' < a$ and $a \leq b' < b$.
\item
$a < a' \leq b$ and $b' > b$.
\end{enumerate}
If $C$ is a cut in $\mathcal{S}_0$ of the first type above
then the corresponding cut in $\mathcal{S}$ is $C' = \emptyset$.
Now, let $C_1,\ldots,C_s$ be the cuts in $\mathcal{S}_0$ of the second type,
according to their order in $\mathcal{S}_0$ (note that $s$ can be zero).
Denote $C_l = \vcut{a_l}{b_l}{j}$.
Since the application of $C_l$ removes all buttons in $B[[a_l,b_l],\{j\}]$,
we have that $a \leq b_1 < b_2 < \cdots < b_s$.
Let $i_1 < \cdots < i_k$ be the row indices of the top $k$ buttons in
$B[X,\{j\}]$.
For each cut $C_l$, the corresponding cut in $\mathcal{S}$ is
$C'_l = \vcut{a_l}{i_l}{j}$.
The cuts of the third type are handled analogously:
Let $\hat{C}_1,\ldots,\hat{C}_t$ be the cuts $\mathcal{S}_0$ of the third type,
where $\hat{C}_l = \vcut{\hat{a}_l}{\hat{b}_l}{j}$.
Let $\hat{\imath}_1 > \cdots > \hat{\imath}_k$ be the row indices of the bottom
$k$ buttons in $B[X,\{j\}]$.
For each cut $\hat{C}_l$, the corresponding cut in $\mathcal{S}$ is
$\hat{C}'_l = \vcut{\hat{\imath}_l}{\hat{b}_l}{j}$.
Finally, the cut that corresponds to $C^*$ is
$\vcut{\min(a^*,i_s)}{\max(b^*,\hat{\imath}_t)}{j}$.

It is easy to show by induction that the cuts in $\mathcal{S}$ are valid cuts
(namely, every cut contains buttons of the same color, and the endpoints of the
cut contains buttons) and that the application of the cuts of $\mathcal{S}$
removes all the buttons of $B$.
Therefore, $S$ is a solution for $(B,k)$.
\end{proof}

Suppose that $(B,k)$ is a yes instance.
Let $\mathcal{S}$ be the solution of Claim~\ref{clm:S} for $(B,k)$.
Since there is no cut in $\mathcal{S}$ with endpoint $(i_0,j)$ for some $j$,
$\mathcal{S}$ is also a solution for $(B',k)$. Therefore, $(B',k)$ is a yes
instance.

The proof of the opposite direction is similar.
Suppose that $(B',k)$ is a yes instance.
Using the same arguments used to prove Claim~\ref{clm:S}, we have that
there is a solution $\mathcal{S}$ for $(B',k)$
such that no cut of $\mathcal{S}$ is contained in a row of $X$.
We now generate a solution $\mathcal{S}_2$ for $(B,k)$
by generating a cut $C'$ for every cut $C$ in $\mathcal{S}$.
Let $j$ be a column of $B$ such that $B[i_0,j] \neq 0$.
Let $i_1$ be the minimum index such that $i_1 > i_0$ and $B[i_1,j] \neq 0$.
Let $C_j = \vcut{a_j}{b_j}{j}$ be the cut in $\mathcal{S}$ whose application
removes the button $B[i_1,j]$.
If $a_j = i_1$ then the corresponding cut in
$\mathcal{S}_2$ is $\vcut{i_0}{b_j}{j}$.
Otherwise, the corresponding cut in $\mathcal{S}_2$ is $C'_j = C_j$.
For a cut $C$ in $\mathcal{S}$ that is not one of the cuts $C_j$ above,
the corresponding cut in $\mathcal{S}_2$ is $C' = C$.
It is easy to verify that $\mathcal{S}_2$ is a solution for $(B,k)$.
Therefore, $(B,k)$ is a yes instance.
\end{proof}

\begin{rrule}
If the number of rows is more than $(4k^2+1)(k+1)k (4k+6)^k$
return no.\label{rule:num-rows}
\end{rrule}
\begin{lemma}
Rule~(\currentrule) is safe.
\end{lemma}
\begin{proof}
Suppose that $(B,k)$ is yes instance and let $\mathcal{S}$ be a solution for
$(B,k)$.

We first show that the rows of $B$ can be partitioned into at most $4k^2+1$
row blocks.
Define $I_0$ to be a set containing $i$ and $i+1$ for every heavy row $i$.
Let $j$ be some column. If $j$ is a light column, define
$I_j = \{i,i+1 \colon B[i,j] \neq 0\}$.
If $j$ is a heavy column, let $I_j$ be the set of all indices $i$ such that
$B[i,j] \neq 0$ and $B[i,j] \neq B[i',j]$, where $i'$ is the maximum index
such that $i' < i$ and $B[i',j] \neq 0$ (if the index $i'$ does not exist,
the index $i$ is not in $I_j$).
Let $I = \bigcup_{j=0}^m I_j$.
We partition $[1,n]$ into row blocks according to $I$ as follows.
Suppose that $I = \{i_1,\ldots i_p\}$ where $i_1 < i_2 < \cdots < i_p$ and
denote $i_{p+1} = n$.
The row blocks are $[i_l,i_{l+1}-1]$ for every $l \leq p$.
Additionally, if $i_1 > 1$ we define the row block $[1, i_1-1]$.
It is clear from the definition of $I$ that each generated interval is
a row block.

We now give an upper bound on $|I|$.
The heavy rows generate at most $2k$ elements in $I$
(since there are at most $k$ heavy rows).
We next consider the contribution of the light columns to $I$.
Since we already bounded the number of indices generated by the heavy rows,
we only need to consider the light buttons in the light columns.
These buttons generate at most $2k^2$ elements in $I$
(since there are at most $k^2$ light buttons).

We now consider some heavy column $j$.
Suppose we remove buttons from $B$ according to the cuts of $\mathcal{S}$,
and after each removal we update $I_j$.
Each cut $C$ in $\mathcal{S}$ (except the last) can decrease the size of $I_j$
by at most 2:
The cut $C$ can remove all the buttons in $B[[i_l,i_{l+1}-1],\{j\}]$ for some
specific $l$. This causes the removal of $i_l$ from $I_j$.
Additionally, if $B[i_{l-1},j] = B[i_{l+1},j]$ then $i_{l+1}$ will also be
removed from $I_j$.
After all the cuts of $\mathcal{S}$ are applied, $I_j = \emptyset$.
It follows that before applying the cuts, $|I_j| \leq 2(k-1)$.
Therefore, column $j$ generates at most $2k$ elements in $I$.
Since there are at most $k$ heavy columns, we conclude that
$|I| \leq 2k+2k^2+2(k-1)k = 4k^2$.
Therefore, the number of generated row blocks is at most $4k^2+1$.

We now partition the row blocks generated above to sub-blocks.
We perform a recursive procedure on every block $X_0$ generated above whose
size is at least 2.
Note that for such block, $B[X_0,[1,m]]$ contains buttons only in heavy columns
and therefore $|J_{X_0}| \leq k$.
Denote by $X = [a,b]$ the current block given to the recursive procedure.
The recursion stops when $J'_X = J_X$.
Note that in this case, $|X| \leq (k+1)k$ since every row of $X$ contains
at least one button (otherwise Rule~(\ref{rule:empty}) can be applied)
and by the definition of $J'_X$, the number of buttons in
$B[X,[1,m]]$ is at most $(k+1)|J_X| \leq (k+1)k$
($|J_X| \leq k$ since every column $j \in J_X$ is heavy).

Now suppose that $J'_X \neq J_X$.
We consider two cases.
In Case~1, $J'_X \neq \emptyset$.
In this case, arbitrarily pick $j \in J'_X$.
Let $i_1 < i_2 < \cdots < i_p$ be the row numbers of the buttons
in $B[X,\{j\}]$, and define $i_0 = 0$ and $i_{p+1} = n+1$.
We partition $X$ into sub-blocks $\{i_l\}$ for every $1 \leq l \leq p$,
and $[i_l+1,i_{l+1}-1]$ for every $0 \leq l \leq p$ such that $i_{l+1}>i_l+1$.
We then continue recursively on each of these sub-blocks.
Note that the number of sub-blocks generated from $X$ is at most $2k+3$.
Moreover, for every generated sub-block $X'$ of size at least 2 we have
$j \notin J_{X'}$ and therefore $|J_{X'}| \leq |J_X|-1$.
Since $|J_{X_0}| \leq k$ for the initial block $X_0$,
it follows that the number of times Case~1 can occur during a chain of
recursive calls is at most $k$.

In Case~2, $J'_X = \emptyset$.
For every $j \in J_X$, let $i_j$ be the row number of the $(k+1)$-th
button from the top in $B[X,\{j\}]$.
Let $i = \max\{i_j \colon j \in J_X\}$ and let $j^*$ be the column such that
$i = i_{j^*}$.
We create two sub-blocks $X_1 = [a,i]$ and $X_2 = [i+1,b]$ and continue
recursively on each of these sub-blocks.
Note that $j^* \in J'_{X_1}$ and in particular $J'_{X_1} \neq \emptyset$.
Additionally, there is at least one column $j \in J_X$ such that the number of
buttons in $B[X_2,\{j\}]$ is less than $k$, since otherwise
Rule~(\ref{rule:row-block}) can be applied on $X$ with $i_0 = i$.
Therefore, $J'_{X_2} \neq \emptyset$.
It follows that the number of times Case~2 can occur during a chain of
recursive calls is bounded by the number of times Case~1 can occur.

From the above, we have that the number of sub-blocks generated from some block
$X_0$ is at most $(4k+6)^k$.
Therefore, the number of rows in $B$ is at most
$(4k^{2}+1)(4k+6)^k (k+1)k$.
It follows that Rule~(\currentrule) is safe.
\end{proof}

\refstepcounter{rrule}
\label{rule:column-block}
\refstepcounter{rrule}
\label{rule:num-columns}

We also use additional reduction rules,
Rule~(\ref{rule:column-block}) and Rule~(\ref{rule:num-columns}),
that are analogous to
Rule~(\ref{rule:row-block}) and Rule~(\ref{rule:num-rows}), respectively.
If these rules cannot be applied, the number of columns in $B$ is at most
$(4k^2+1)(k+1)k (4k+6)^k$.

\begin{rrule}
If the number of buttons is at least $k \cdot \max(n,m)+1$,
return no.\label{rule:num-buttons}
\end{rrule}

The safeness of Rule~(\currentrule) is proved in~\cite{agrawal2016kernelizing}.
The following lemma is also proved in~\cite{agrawal2016kernelizing}.
\begin{lemma}\label{lem:alg}
Let $(B,k)$ be an instance of orthogonal buttons and scissors such that $B$
contains $l$ buttons.
Then, the instance can be solved in $l^{2k} (n+m)^{O(1)}$ time.
\end{lemma}

Let $(B,k)$ be an instance on which the above reduction rules cannot be applied.
Since Rule~(\ref{rule:num-rows}), Rule~(\ref{rule:num-columns}),
and Rule~(\ref{rule:num-buttons}) cannot be applied,
the number of buttons in $B$ is $2^{O(k\log k)}$.
By Lemma~\ref{lem:alg}, the instance $(B,k)$ can be solved
in $2^{O(k^2 \log k)} (n+m)^{O(1)}$ time.

\bibliographystyle{abbrv}
\bibliography{buttons}

\end{document}